\documentclass[12pt,reqno]{article}
\usepackage[letterpaper,margin=1.6cm]{geometry}
\usepackage{latexsym, amsfonts, amsthm,amssymb,amscd,amsmath,makeidx}
\usepackage{enumerate}
\usepackage[normalem]{ulem}
\usepackage{exscale}
\usepackage{overpic} 
\usepackage{color} 
\usepackage{graphicx}
\usepackage{overpic}  
\usepackage{bm}
\usepackage{comment}
\usepackage{caption}
\usepackage{float}
\usepackage{array} 
\usepackage{booktabs} 
\usepackage{tabularx}
\usepackage{mathtools}
\usepackage{subcaption}
\usepackage{mathrsfs}
\usepackage{amssymb}
\usepackage{mathtools}

\usepackage{tikz}
\usepackage{blindtext}
\usepackage{bbm}
\usepackage{fancyhdr} 
\usepackage{mathrsfs}
\usepackage{wrapfig}
\usepackage{hyperref}
\usepackage{bookmark}
\usepackage{cleveref}
\usepackage{cases}
\usepackage{multirow}
\usepackage{pdflscape}
\usepackage[title]{appendix}

\hypersetup{
	colorlinks,
	citecolor=red,
	filecolor=black,
	linkcolor=blue,
	urlcolor=black
}

\definecolor{DarkGreen}{rgb}{0.2,0.6,0.2}

\def\om{\omega}
\def\trace#1{\text{trace}\, #1}

\newcommand{\norm}[1]{\left\lVert {#1}\right\rVert}
\newcolumntype{Y}{>{\centering\arraybackslash}X}

\def\ua{\uparrow}
\def\da{\downarrow}

\def\wh{\widehat}
\def\wt{\widetilde}

\newcolumntype{C}{>{\centering\arraybackslash}X}

\def\ignore#1{}

\newcommand\scalemath[2]{\scalebox{#1}{\mbox{\ensuremath{\displaystyle #2}}}}

\numberwithin{equation}{section}

\usetikzlibrary{positioning,calc,cd}

\allowdisplaybreaks[4]

\def\cR{{\mathscr R}}

\newtheorem{theorem}{Theorem}[section]
\newtheorem{proposition}[theorem]{Proposition}
\newtheorem{lemma}[theorem]{Lemma}

\theoremstyle{definition}
\newtheorem{definition}[theorem]{Definition}

\newtheorem{remark}[theorem]{Remark}

\def\<{\langle}
\def\>{\rangle}
\def\wt#1{\widetilde{#1}}

\def\ua{\uparrow}
\def\da{\downarrow}

\def\wh{\widehat}
\def\wt{\widetilde}

\def\argmin{\mathop{\hbox{\rm arg\,min}}}

\newcommand{\bR}{\mathbb{R}}

\newcommand{\bN}{\mathbb{N}}
\newcommand{\bP}{\mathbb{P}}

\newcommand{\bT}{\mathbb{T}}

\begin{document}
	\title{On the rate of convergence of estimating the Hurst parameter of rough stochastic volatility models}
\author{Xiyue Han$^*$ and Alexander Schied\thanks{
		Department of Statistics and Actuarial Science, University of Waterloo, 200 University Ave W, Waterloo, Ontario, N2L 3G1, Canada. E-Mails: {\tt xiyue.han@uwaterloo.ca, aschied@uwaterloo.ca}.\hfill\break
		The authors gratefully acknowledge support from the Natural Sciences and
		Engineering Research Council of Canada through grant
		RGPIN-2024-03761. }}
	\date{\normalsize First version: April 15, 2025\\
	\normalsize This version: \today}
\maketitle

\begin{abstract}
	In [Han \& Schied, 2023, \textit{arXiv 2307.02582}], an easily computable scale-invariant estimator $\cR^s_n$ was constructed to estimate the Hurst parameter of the drifted fractional Brownian motion $X$ from its antiderivative. This paper extends this result by proving that $\cR^s_n$ also consistently estimates the Hurst parameter when applied to the antiderivative of $g \circ X$ for a general nonlinear function $g$. We also establish an almost sure rate of convergence in this general setting.  Our result applies, in particular, to the estimation of the Hurst parameter of a wide class of rough stochastic volatility models from discrete observations of the integrated variance, including the rough fractional stochastic volatility model.
	
\end{abstract}

\section{Introduction and statement of the main result}
We consider a stochastic volatility model, where the price process is driven by a stochastic differential equation with respect to a standard Brownian motion $B$, 
\begin{equation*}
	dS_t=\sigma_tS_t\,dB_t.
\end{equation*}
Here, the process $\sigma$ is continuous and adapted and referred to as the volatility process. It was discovered empirically by Gatheral, Jaisson and Rosenbaum \cite{GatheralRosenbaum} that the volatility process $\sigma$ does not exhibit diffusive behavior but instead is much rougher. This discovery led to the development of rough stochastic volatility models, in which the smooth diffusive dynamics of classical models are replaced by rougher counterparts, such as fractional Brownian motion or Gaussian Volterra processes. A specific example here is the \textit{rough fractional stochastic volatility model} proposed in \cite[Section 3]{GatheralRosenbaum}, where logarithmic volatility is modeled by a fractional Ornstein--Uhlenbeck process $X^H$. To be more precise, we have
\begin{equation}\label{log volatility eq}
	\sigma_t=\exp(X^H_t),
\end{equation}
where $X^H$ solves the following integral equation
\begin{equation}\label{OU process eq}
	X^H_t=x_0+\rho\int_0^t(\mu-X^H_s)\,ds+W^H_t,\qquad t\ge0,
\end{equation}
for a fractional Brownian motion $W^H$ with Hurst parameter $H\in(0,1/2)$. In particular, it was shown in \cite{Gatheral2018rough} that a value $H \approx 0.1$ appears to be most adequate for capturing the stylized facts of empirical volatility time series. Following the advent of the rough fractional stochastic volatility, many rough volatility models have been proposed. Notable examples include the rough Heston model \cite{EuchRosenbaumGatheral2019roughening, EuchRosenbaumRoughHeston18} and the rough Bergomi model \cite{BayerGatheralFriz16, Forde2022} in which the volatility process is based on Gaussian Volterra processes. For an overview of recent developments in rough volatility, we refer to the book \cite{RoughvolBook}. 

In the model \eqref{log volatility eq}, the degree of roughness of the volatility process is governed by the Hurst parameter $H$. When trying to estimate the Hurst parameter $H$, the major difficulty is that the volatility process $\sigma$ cannot be observed directly from given data; only the asset prices $S$ can be observed. Thus, one typically computes the quadratic variation of the log prices, 
\begin{equation}\label{eq integrated variance}
	\<\log S\>_t = \int_{0}^t \sigma_s^2 \,ds,
\end{equation}
which is also known as the integrated variance, and then obtains proxy values $\hat \sigma_t$ by numerical differentiation. The roughness estimation is then based on those proxy values. However, it was shown in \cite{ContDasArtefact} that the error that arises in the numerical differentiation might distort the final roughness estimation. 

Various approaches have been proposed to tackle this issue. Assuming that the volatility process is driven by a fractional Brownian motion, Bolko et al.~\cite{BolkoPakkanen2022GMM} employ the generalized method of moments to estimate the Hurst parameter. In addition, Fukasawa et al.~\cite{Fukasawa2022Estimation} develop a Whittle-type estimator under a similar parametric setting. Chong et al.~\cite{Chong2022CLT, Chong2022Minimax} substantially extend the previous results by considering a semi-parametric setup, in which, with the exception of the Hurst parameter of the underlying fractional Brownian motion, all components are fully non-parametric.

To address the same issue,  in \cite{HanSchiedDerivative} we constructed an estimator $\wh \cR_n$ that estimates the Hurst parameter $H$ in \eqref{log volatility eq} based on discrete observations of the integrated variance \eqref{eq integrated variance}. In contrast to other estimation schemes \cite{ BolkoPakkanen2022GMM,Chong2022CLT, Chong2022Minimax, Fukasawa2022Estimation}, our estimator is constructed in a strictly pathwise setting and, in fact, estimates the so-called roughness exponent $R$, which coincides with the Hurst parameter for fractional Brownian motion \cite{HanSchiedHurst}. The fact that our estimator is built on a strictly pathwise approach makes it very versatile and applicable also in situations where trajectories are not based on fractional Brownian motion; see, e.g., \cite[Examples 3.5]{HanSchiedDerivative} for further discussions.

In our pathwise setting, we consider a given but unknown function $x \in C[0,1]$ and denote $y(t) := \int_{0}^{t}g(x(s))\,ds$ for a function $g \in C^2(\bR)$. Based on the observations of function $y$ over the dyadic partition $\bT_{n+2}$, i.e., $\{y(k2^{-n-2}):k = 0, \cdots, 2^{n+2}\}$, we introduce the coefficients
\begin{equation}\label{eq vartheta}
	\vartheta_{n, k}:=2^{3 n / 2+3}\left(y\left(\frac{4 k}{2^{n+2}}\right)-2 y\left(\frac{4 k+1}{2^{n+2}}\right)+2 y\left(\frac{4 k+3}{2^{n+2}}\right)-y\left(\frac{4 k+4}{2^{n+2}}\right)\right),
\end{equation}
for $0 \le k \le 2^n-1$. Our estimator for the roughness exponent of the trajectory $x$ is now given by 
\begin{equation}\label{wh Rn eq}
	\wh\cR_n(y): = 1-\frac{1}{n}\log_2 \sqrt{\sum_{k = 0}^{2^{n}-1}\vartheta_{n,k}^2}.
\end{equation}
In contrast to many other estimators proposed in the literature, $\wh\cR_n(y)$ can be computed in a straightforward manner. For instance, to estimate the Hurst parameter of the rough fractional stochastic volatility model \eqref{log volatility eq}---\eqref{OU process eq}, we take $x$ as a sample trajectory of the fractional Ornstein--Uhlenbeck process $X^H$ and $g(t) = e^{2t}$ so that 
\begin{equation}\label{y with g=e2t eq}
	y(t) = \int_0^t \big(\exp(x(s))\big)^2 \,ds, \qquad 0 \le t \le 1,
\end{equation}
replicates the integrated variance corresponding to the sample trajectory $x$.

In particular, the fact that our estimator $\wh \cR_n$ is derived from a purely pathwise consideration makes it applicable to an even wider class of processes. For instance, let $X$ be given by 	\begin{equation}\label{XH eq}
	X_t:=x_0+W^H_t+\int_0^t\xi_s\,ds,\qquad0\le t\le 1,
\end{equation}
where $\xi$ is progressively measurable with respect to the natural filtration of $W^H$ and satisfies the following additional assumption.
\begin{itemize}
	\item If $H<1/2$, we assume  that the function $t\mapsto\xi_t$ is $\bP$-a.s.~bounded in the sense that  there exists a finite random variable $C$ such that $|\xi_t(\om)|\le C(\om)$ for a.e.~$t$ and $\bP$-a.s.~$\om$. 
	\item If $H>1/2$, we assume that for $\bP$-a.s.~$\om$ the function $t\mapsto \xi_t(\om)$ is H\"older continuous with some exponent $\alpha(\om)>2H-1$.\end{itemize}

We emphasize that the class of processes defined by~\eqref{XH eq} is sufficiently general. In particular, for $H < 1/2$, any $\bP$-a.s continuous drift $\xi$ clearly satisfies the condition. One can also specify conditions on the drift term of a stochastic differential equation driven  by fractional Brownian motion under which the assumption is satisfied; see Theorem 1.6 in \cite{HanSchiedGirsanov}. Now, suppose that $g \in C^2(\bR)$ is strictly monotone and let $Y_t = \int_{0}^{t}g( X_s)\,ds$. It is shown in \cite[Corollary 2.3]{HanSchiedDerivative} that $\wh \cR_n(Y) \rightarrow H$ as $n \ua \infty$ with probability one. This consistency result applies in particular to the rough fractional stochastic volatility model defined in \eqref{log volatility eq} and \eqref{OU process eq}, where we take $\xi_s = \rho(\mu - X^H_s)$ and $g(t) = e^{2t}$.

However, the estimator $\wh \cR_n$ is  not scale-invariant, and the performance of $\wh \cR_n$ is highly sensitive to the underlying scale of the function $y$. To solve this issue, a scale-invariant modification of $\wh\cR_n$ was constructed in \cite{HanSchiedDerivative} as in the following definition.

\begin{definition}\label{sls and tls def}Fix $m\in\bN$ and $\alpha_0,\dots,\alpha_{m}\ge0$ with $\alpha_0>0$.
	For $n>m$, the \textit{sequential scaling factor} $\eta_{n}^{s}$ and the \textit{sequential scale estimate} $\cR^s_{n}(y)$ are defined as follows,
	\begin{equation}\label{eq_def_seqloc}
		\begin{split}
			\eta^{s}_{n}&:= \argmin_{\eta > 0} \sum_{k = n-m}^{n}\alpha_{n-k}\Big(\wh\cR_k(\eta y) - \wh\cR_{k-1}(\eta y)\Big)^2 \quad \text{and} \quad
			\cR^s_{n}(y):= \wh\cR_{n}(\eta^{s}_{n} y).
		\end{split}
	\end{equation}
	The corresponding mapping $\cR^s_{n}:C[0,1]\rightarrow \bR$ is called the \textit{sequential scale estimator}.
\end{definition}

There is no rule of thumb for choosing the parameters $\alpha_0, \dots, \alpha_m$. As a matter of fact, the performance of $\cR^s_n$ is dependent on the actual Hurst parameter $H$. Nevertheless, as will be shown in \eqref{eq rate of convergence Y}, regardless of the choice of $\alpha_0, \dots, \alpha_m$, the sequential scale estimator $\cR^s_n$ shares the same asymptotic rate of convergence; see also \cite{HanSchiedHurst} for further approaches to construct scale-invariant estimators. For given $\alpha_0, \dots, \alpha_m$, the sequential scale estimator $\cR^s_n$ can be represented as a linear combination of the estimators $\wh\cR_k$ as follows,
\begin{equation*}
	\cR_n^s=\beta_{n, n} \widehat{\mathscr{R}}_n+\beta_{n, n-1} \widehat{\mathscr{R}}_{n-1}+\cdots+\beta_{n, n-m-1} \widehat{\mathscr{R}}_{n-m-1},
\end{equation*}
where the coefficients $\beta_{n,k}$ are explicitly given in \cite[Proposition 2.6]{HanSchiedDerivative}. To be more specific, we have
$$\beta_{n,k}=\begin{cases}\displaystyle1+\frac{\alpha_0}{c^{\textrm{\rm s}}_{n}n^2(n-1)}&\text{if $k=n$,}\\
	\displaystyle\frac1{c^{\textrm{\rm s}}_{n}nk}\Big(\frac{\alpha_{n-k}}{k-1}-\frac{\alpha_{n-k-1}}{k+1}\Big)&\text{if $n-m\le k\le n-1$,}\\
	\displaystyle
	\frac{-\alpha_m}{c^{\textrm{\rm s}}_{n}n(n-m)(n-m-1)}&\text{if $ k= n-m-1$,}
\end{cases}\quad\text{for}\quad c^{s}_{n}:= \sum_{k = n-m}^{n}\frac{\alpha_{n-k}}{k^{2}(k-1)^{2}}.
$$Therefore, the sequential scale estimator $\cR_n^s$ can be computed in a fast and straightforward manner. In addition, the rate of convergence of the sequential scale estimator $\cR^s_n$ was primarily studied in \cite[Theorem 2.7]{HanSchiedDerivative}, which is quoted here for the convenience of the reader. Let $X$ be as in \eqref{XH eq} and $Y_t = \int_{0}^{t} X_s\,ds$. Then the following almost sure rate of convergence holds for the sequential scale estimator $\cR^s_n$,
\begin{equation}\label{eq rate of convergence Y}
	|\cR^s_n(Y) - H| = \mathcal{O}\big(2^{-n/2}\sqrt{\log n}\big).
\end{equation}
Here, the rate of convergence was only established under the assumption that $g$ is the identity function. Hence, this result cannot be applied directly to establish the consistency or the convergence rate of $\cR^s_n$ for rough stochastic volatility models. In these models, we typically make discrete observations of the integrated variance of the form
$$\<\log S\>_t = \int_0^t \sigma^2_s\,ds=\int_0^t g(X_s)\,ds,\qquad t \in [0,1].
$$
for some strictly increasing nonlinear function $g \in C^2(\bR)$. Such a choice leads to non-Gaussian dynamics and thus lies beyond the scope of \cite[Theorem 2.7]{HanSchiedDerivative}. Our following theorem extends the convergence result \eqref{eq rate of convergence Y} to the case in which  $g$  twice continuously differentiable function satisfying a very mild regularity condition.

\begin{theorem}\label{thm main}
	Suppose that $g \in C^2(\bR)$ satisfies\begin{equation}\label{eq positive}
		\int_0^1 \big(g’(X(s))\big)^2 \, ds > 0 \qquad \bP\text{-a.s.}
	\end{equation}
	Let $m \in \bN$,  $\alpha_0 > 0$ and $\alpha_1,\dots,\alpha_m \ge 0$. Let $X$ be as in \eqref{XH eq} and 	\begin{equation*}
		Y_t = \int_0^t g(X_s) \,ds, \qquad t \in [0,1].
	\end{equation*}
	Then the following almost sure rate of convergence holds for the sequential scale estimator $\cR^s_n$, 
	\begin{equation}\label{eq rate}
		|\cR^s_n(Y) - H| = \mathcal{O}\left(\sqrt{n}\cdot2^{-(\frac{H}{2} \wedge \frac{1}{4})n}\right).
	\end{equation}
\end{theorem}

\begin{remark}
	Note that the condition \eqref{eq positive} is automatically satisfied if $g$ is strictly monotone, such as the function  $g(t) = e^{2t}$  used in \eqref{y with g=e2t eq}. It is also satisfied for the choice $g(t)=t^2$. Indeed, we have 
	$\int_0^1\big(g'(W^H_s)\big)^2\,ds=4\int_0^1(W^H_s)^2\,ds>0$ $\bP$-a.s.,
	because $\{(s,\om):W^H_s(\om)=0\}$ is a $\rm{Leb}[0,1]\otimes\bP$-null set, and so \eqref{eq positive}  follows by way of the absolute continuity of the law of $X$ established in \cite{HanSchiedGirsanov}. 
\end{remark}

It is also worthwhile to point out that the proof of \cite[Theorem 2.7]{HanSchiedDerivative} relies essentially on the Gaussianity of the antiderivative $Y$. Hence, its approach does not extend to the setting of \Cref{thm main}, where the process $g \circ X$ is no longer Gaussian, and so neither is the process $Y$. Instead, \Cref{thm main} is established by a pathwise approach, which is, in fact, robust and applicable to a wide range of rough volatility models. As will become clear in the proof, the convergence rate in \Cref{thm main} depends only on the convergence rate of $\cR^s_n(Y)$ (see \eqref{eq rate of convergence Y}) and the Hölder continuity of the process $X$, or equivalently, of the fractional Brownian motion $W^H$. In many rough volatility models, such as the rough Bergomi model \cite{Gatheral2018rough}, the process $X$ is modeled by a Gaussian Volterra process, wheares the function $g$ remains to be exponential. In this setting, the convergence rate in \eqref{eq rate of convergence Y} can be deduced by exploring the covariance structure of $Y$; see, e.g., \cite[Proposition 2.7]{HanSchiedDerivative}.

A potential concern is that the convergence rate of our estimator depends on the Hurst parameter $H$ itself. In particular, for very small values of $H$, such as $H \approx 0.1$, commonly considered in the rough volatility literature, large sample sizes might be required for an accurate estimation. In contrast, the following numerical study shows that the finite-sample performance at $H = 0.1$ is in fact better then the asymptotic rate \eqref{eq rate} would suggest. Even at $n = 10$, or equivalently, with $2^{12}$ observations, the sequential scale estimator $\cR^s_n$ already gives very reliable estimates. 

\begin{figure}[H]
	\centering
	\includegraphics[width=7.5cm]{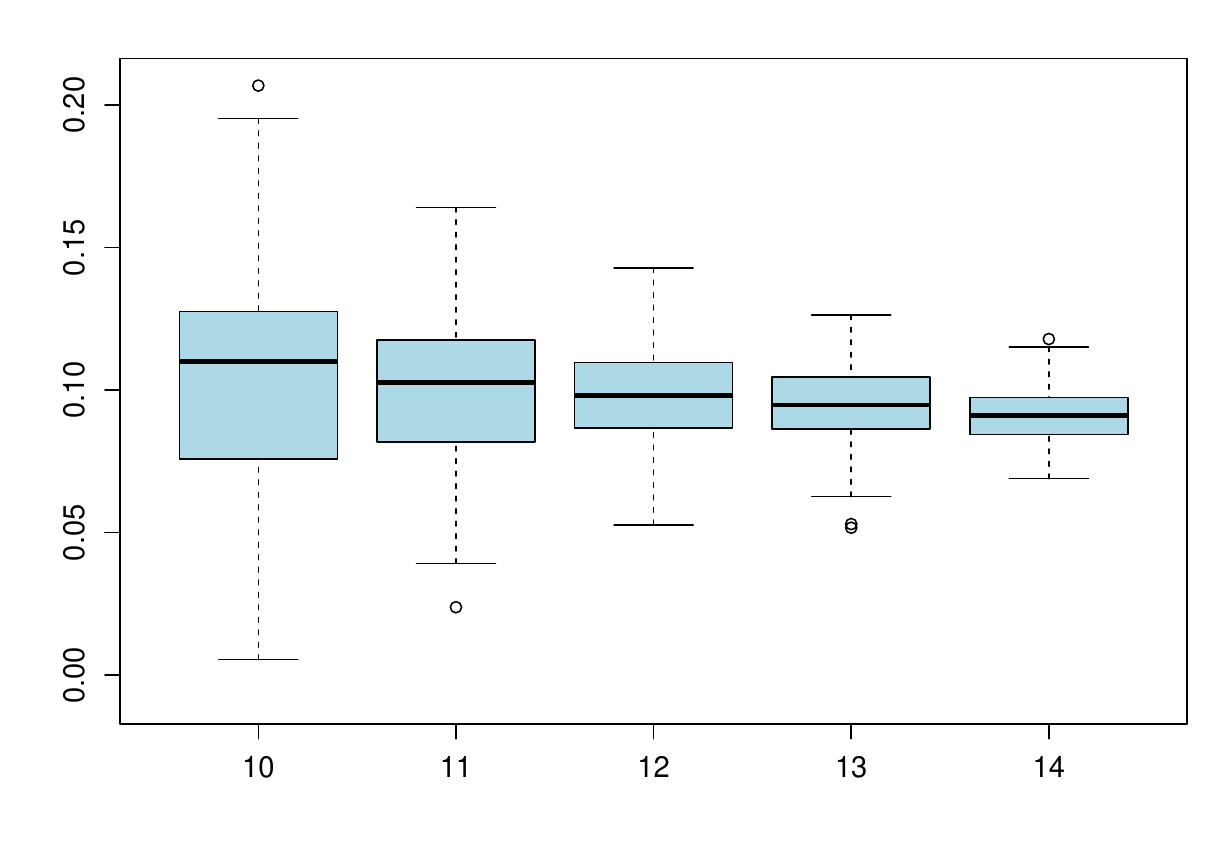}
	\,	\includegraphics[width=7.5cm]{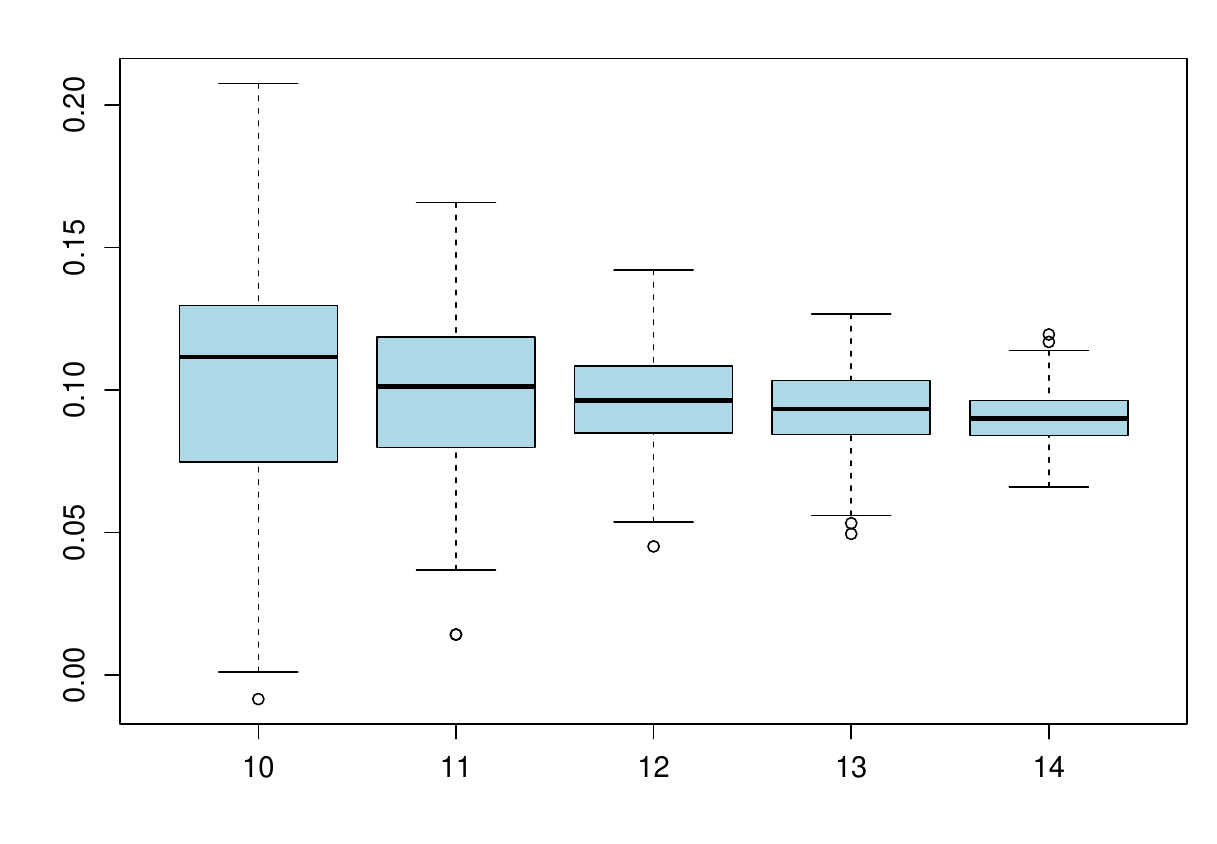}
	\caption{Box plots of the sequential scale estimates $ \cR^s_n(Y)$ for $n=10,\dots, 14$, based on 200 sample paths of fractional Ornstein--Uhlenbeck process $X^H$ with $g(t) = \exp(2t)$ (left) and $g(t) = (t- 2)^2 + \sin(2\pi t)$ (right). The other parameters are chosen to be $x_0 = 2$, $\rho = 0.2$, $\mu = 2$, $m = 3$ and $\alpha_k = 1$ for $k = 0,1,2,3$.} 
	\label{fig:Box Ori si}
\end{figure}

\section{Proof of Theorem 1.2}\label{sec proof}
To prove \Cref{thm main}, we adopt the notation used in \cite{HanSchiedDerivative}. For any given deterministic function or random process $f$, we write 
\begin{equation}\label{coefficients depending on f eq}
	\begin{split}
		\theta^f_{n,k}& = 2^{n/2}\left(2f\Big(\frac{2k+1}{2^{n+1}}\Big)-f\Big(\frac{k}{2^{n}}\Big) - f\Big(\frac{k+1}{2^{n}}\Big)\right), \\
		\vartheta^f_{n,k}&= 2^{3n/2+3}\left(f\Big(\frac{4k}{2^{n+2}}\Big)-2f\Big(\frac{4k+1}{2^{n+2}}\Big)+2f\Big(\frac{4k+3}{2^{n+2}}\Big)-f\Big(\frac{4k+4}{2^{n+2}}\Big)\right).
	\end{split}
\end{equation}
Here, the coefficients $(\theta^f_{n,k})$ are also referred to as the Faber--Schauder coefficients of $f$, and coefficients $(\vartheta^f_{n,k})$ are the approximated Faber--Schauder coefficients \eqref{eq vartheta} with respect to $f$ as in \cite[Theorem 2.1]{HanSchiedMatrix}. Furthermore, for a given function or process $f$, we write $\bar{\bm \vartheta}^f_n := \big(\vartheta^f_{n,0},\vartheta^f_{n,1},\dots,\vartheta^f_{n,2^n-1}\big)^\top$. In particular, if $f$ is a Gaussian process, $\bar{\bm \vartheta}^f_n$ then defines a Gaussian vector. In this section, we let
\begin{equation*}
	Y_t := \int_0^t W^H_s\,ds \quad \text{and} \quad V_t = \int_{0}^{t} g\big(W^H_s\big) \,ds
\end{equation*}
be the antiderivative of $W^H$ and $g \circ W^H$ respectively. We denote the covariance matrix of the Gaussian vector $\bar{\bm \vartheta}^{Y}_n$ by $\Psi_n$, and we fix $\tau_H := \trace{\Psi_0}$. It was shown in \cite[Proposition 4.9]{HanSchiedDerivative} that there exists a positive constant $c_H > 0$ such that
\begin{equation}\label{eq old convergence}
	\limsup_{n \ua \infty}\delta_n^{-1}\left|2^{n(H-1)}\norm{\frac{\bar{\bm \vartheta}^Y_n}{\sqrt{\tau_H}}}_{\ell_2} - 1\right| \le 1 \qquad \bP\text{-}a.s.
\end{equation}
for $\delta_n = c_H\cdot2^{-n/2}\sqrt{\log n}$. For $n \in \bN$, we now denote 
\begin{equation*}
	\bar {\bm \vartheta}^Y_{2n,i}:= \left(\vartheta_{2n,2^ni},\vartheta_{2n,2^ni+1},\cdots, \vartheta_{2n,2^n(i+1)-1}\right)^\top, \qquad 0 \le i \le 2^n-1.
\end{equation*}
In other words, the vectors $\left(\bar {\bm \vartheta}^Y_{2n,i}\right)$ divide the Gaussian vector $\bar{\bm \vartheta}^Y_{2n}$ into $2^n$ equally partitioned subvectors.

The proof of \Cref{thm main} results from a sequence of intermediate lemmas, which we summarize below to outline the roadmap of this proof.
\begin{itemize}
	\item First, \Cref{lemma uniform} obtains the uniform almost sure rate of convergence of the Gaussian vectors $(\bar {\bm \vartheta}^Y_{2n,i})$. Furthermore, \Cref{remark pathwise} transfers this convergence result to two-sided bounds for the $\ell_2$-norm of $\bar {\bm \vartheta}^Y_{2n,i}$ in a strictly pathwise sense.
	\item Second, \Cref{lemma bound} shows that the two-sided bounds in \Cref{remark pathwise} would carry over from the $\ell_2$-norm of $\bar {\bm \vartheta}^Y_{2n,i}$ to that of $\bar {\bm \vartheta}^V_{2n,i}$.
	\item Next, \Cref{lemma bound 2} derives the two-sided bounds for the $\ell_2$-norm of $\bar {\bm \vartheta}^V_{2n}$ based on the bounds in \Cref{lemma bound}.
	\item Finally, \Cref{lemma convergence} applies the two-sided bounds in \Cref{lemma bound 2} to obtain the almost sure convergence rate of $\cR_{2n}^s(V)$.
\end{itemize}

Let us begin by considering the uniform convergence rate of the Gaussian vectors $\left(\bar {\bm \vartheta}^Y_{2n,i}\right)$ similar to \eqref{eq old convergence}. 
\begin{lemma}\label{lemma uniform}
	For $H \in (0,1)$, there exists a constant $c_H > 0$ such that 
	\begin{equation*}
		\lim_{n \ua \infty}\delta_n^{-1}\sup_{0 \le i \le 2^n-1}\left|2^{n(2H-3/2)}\norm{\frac{\bar{\bm \vartheta}_{2n,i}}{\sqrt{\tau_H}}}_{\ell_2} - 1\right| \le 1 \qquad \bP\text{-}a.s.
	\end{equation*}
	for $\delta_n = c_H \cdot 2^{-n/2}\sqrt{n\log 2  + 2\log n}$.
\end{lemma}
\begin{proof}
	Let us denote the covariance matrix of the Gaussian vector $\bar{\bm \vartheta}_{2n,i}$ by $\Phi_{2n,i}$. By definition, the matrix $\Phi_{2n,i}$ is the $i^{\text{th}}$ diagonal partitioned matrix of $\Psi_{2n}$. As the fractional Brownian motion $W^H$ is self-similar and admits stationary increments, we have 
	\begin{equation}\label{eq matrix}
		\Phi_{2n,i} =  2^{(1-2H)n}\Psi_n.
	\end{equation}
	This then gives $$\trace{\Phi_{2n,i}} = 2^{(1-2H)n}\trace{\Psi_n} = 2^{(3-4H)n}\tau_H.$$ Furthermore, applying \cite[Proposition 4.9]{HanSchiedDerivative} to \eqref{eq matrix} yields 
	\begin{equation*}
		\norm{\Phi_{2n,i}}_2 =  2^{(1-2H)n}\norm{\Psi_n}_2 \le \kappa_H 2^{n(2-4H)}, \qquad 0 \le i \le 2^n-1,
	\end{equation*}
	for some $\kappa_H > 0$. For any given $\delta > 0$, it follows from the concentration inequality \cite[Lemma 3.1]{BaudoinHairer} that 
	\begin{align*}
		\MoveEqLeft \bP\left(\sup_{0 \le i \le 2^n-1}\left|2^{n(2H-\frac{3}{2})}\norm{\frac{\bar{\bm \vartheta}_{2n,i}}{\sqrt{\tau_H}}}_{\ell_2}-1\right| \ge \delta\right)\\ &= \bP\left(\sup_{0 \le i \le 2^n-1}\left|2^{n(2H-\frac{3}{2})}\norm{\bar{\bm \vartheta}_{2n,i}}_{\ell_2}-\sqrt{\tau_H}\right| \ge \delta\sqrt{\tau_H}\right) \\&= \bP\left(\bigcup_{i = 0}^{2^n-1}\left\{\left|2^{n(2H-\frac{3}{2})}\norm{\bar{\bm \vartheta}_{2n,i}}_{\ell_2}- \sqrt{\tau_H}\right| \ge \delta\sqrt{\tau_H} \right\}\right) \\ &\le \sum_{i = 0}^{2^n-1}\bP\left(\left|2^{n(2H-\frac{3}{2})}\norm{\bar{\bm \vartheta}_{2n,i}}_{\ell_2}- \sqrt{\tau_H}\right| \ge \delta \sqrt{\tau_H}\right)\\&= \sum_{i = 0}^{2^n-1}\bP\left(\left|\norm{\bar{\bm \vartheta}_{2n,i}}_{\ell_2}- \sqrt{\trace{\Phi_{2n,i}}}\right| \ge 2^{n(\frac{3}{2} - 2H)}\delta \sqrt{\tau_H}\right)  \\&\le\sum_{i = 0}^{2^n-1} \phi \exp\left(-\frac{2^{n(3-4H)}\delta^2\tau_H}{4\norm{\Phi_{2n,i}}_2}\right) =2^n\left(\phi  \exp\left(-\frac{2^{n(3-4H)}\delta^2\tau_H}{4\norm{\Phi_{2n,0}}_2}\right)\right) \\&= \phi \exp\left(-\frac{2^{n(3-4H)}\delta^2\tau_H}{4\norm{\Phi_{2n,0}}_2}+n\log2\right) \le \phi \exp\left(\frac{2^{-n}\delta^2\tau_H}{4\kappa_H} + n\log2\right),
	\end{align*}
	for some constant $\phi > 0$. We now take $c_H := \sqrt{4\kappa_H/\tau_H}$ and $\delta_n := c_H\cdot 2^{-n/2}\big(2\log n + n \log 2\big)^{-1/2}$. For each $n \in \bN$, plugging $\delta = \delta_n$ into the above inequality yields that 
	\begin{equation}\label{eq borel cantelli}
		\bP\left(\sup_{0 \le i \le 2^n-1}\left|2^{n(2H-\frac{3}{2})}\norm{\frac{\bar{\bm \vartheta}_{2n,i}}{\sqrt{\tau_H}}}_{\ell_2}-1\right| \ge \delta_n\right) \le \frac{\phi}{n^2}.
	\end{equation}
	As the expression on the right-hand side of \eqref{eq borel cantelli} is absolutely summable, the Borel--Cantelli lemma yields that 
	\begin{equation}\label{eq null set}
		\lim_{n \ua \infty}\delta_n^{-1}\sup_{0 \le i \le 2^n-1}\left|2^{n(2H-\frac{3}{2})}\norm{\frac{\bar{\bm \vartheta}_{2n,i}}{\sqrt{\tau_H}}}_{\ell_2} - 1\right| \le 1 \qquad \bP\text{-}a.s.
	\end{equation}
	This completes the proof.
\end{proof}
We now start our pathwise analysis to obtain the bounds for the $\ell_2$-norm of $\bar {\bm \vartheta}^Y_{2n,i}$. To this end, let us first of all clarify the notation we are going to use in the following proofs. We fix $g \in C^2(\bR)$ and let $x$ be a typical sample path of $W^H$. Here, we refer to the sample paths that do not belong to the null set as typical sample paths. We let
\begin{equation*}
	y(t) := \int_{0}^{t}x(s)\,ds, \qquad	u(t)=g(x(t))\quad\text{and}\quad v(t)=\int_0^tu(s)\,ds=\int_0^tg(x(s))\,ds.
\end{equation*}
Using these notations, we can rephrase \Cref{lemma uniform} in the following strictly pathwise manner. 
\begin{remark}\label{remark pathwise}
	It follows from \Cref{lemma uniform} that for a typical sample path $x$ of fractional Brownian motion $W^H$, there exist $n_x \in \bN$ and a positive constant $c_H > 0$ such that for $n \ge n_x$, we have
	\begin{equation}\label{eq pathwise}
		(1 - \delta_n)^2 \tau_H \le2^{n(4H-3)} \norm{\bar{\bm \vartheta}^y_{2n,i}}^2_{\ell_2} \le (1 + \delta_n)^2 \tau_H,
	\end{equation}
	for $\delta_n = c_H \cdot 2^{-n/2}\sqrt{n\log 2  + 2\log n}$ and all $0 \le i \le 2^n - 1$. Here, the collection of typical sample paths consists of sample paths that are continuous and satisfy the convergence rate \eqref{eq null set}.
\end{remark}

The following lemma shows that the two-sided bounds \eqref{eq pathwise} can be carried over from the sample path $y$ to the sample path $v$.

\begin{lemma}\label{lemma bound}
	Let $x$ be a typical sample path of $W^H$. Then, there exist $n_x \in \bN$, $c_H > 0$ and intermediate times $\tau^a_{n,i}, \tau^b_{n,i} \in [2^{-n}i,2^{-n}(i+1)]$ such that for $n \ge n_x$, we have
	\begin{equation*}
		\begin{split}
			&\left(\big(g'(x(\tau^a_{n,i}))\big)^2(1 - \delta_n)^2-2^{-5Hn/4}\big|g'(x(\tau^b_{n,i}))\big|\right) \tau_H \le2^{n(4H-3)} \norm{\bar{\bm \vartheta}^v_{2n,i}}^2_{\ell_2}\\  & \le \left(\big(g'(x(\tau^b_{n,i}))\big)^2(1 + \delta_n)^2+2^{-5Hn/4}\big|g'(x(\tau^b_{n,i}))\big|\right)  \tau_H,
		\end{split}
	\end{equation*}
	for $\delta_n = c_H \cdot 2^{-n/2}\sqrt{n\log 2  + 2\log n}$ and all $0 \le i \le 2^n - 1$.
\end{lemma}

\begin{proof}
	To prove this lemma, we let $\theta_{m,k}^f(s):=\theta_{m,k}^{f(s+\cdot)}$. That is, $\theta_{m,k}^f(s)$ are the Faber--Schauder coefficients \eqref{coefficients depending on f eq} of the function $t\mapsto f(s+t)$ for given $s\ge0$. One can avoid undefined arguments of functions in case $s+t>1$ by setting $f(t):= f(t \wedge 1)$ for $t \ge 1$. Furthermore, we let 
	\begin{equation*}
		\zeta^x_{n+1,2k}(s):= 2^{(n+1)/2}\left(x\Big(\frac{4k+2}{2^{n+2}}+s\Big)-x\Big(\frac{4k}{2^{n+2}}+s\Big)\right)\big(x(\tau_{n+2,4k}(s))-x(\tau_{n+2,4k+1}(s))\big)
	\end{equation*}
	and 
	\begin{equation*}
		\wt \zeta^x_{n+1,2k}:= 2^{n+5/2}\int_{0}^{2^{-n-1}} \zeta^{x}_{n+1,2k}(s)\,ds,
	\end{equation*}
	where  $\tau_{n+2,k}(s) \in [2^{-n-2}k+s, 2^{-n-2}(k+1)+s]$ are certain intermediate times such  that for $s \in [0,2^{-n-1}]$. It follows from \cite[Equation (4.5)]{HanSchiedDerivative} that for $0 \le k \le 2^{2n}-1$, we have
	\begin{equation}\label{three terms eq}
		\begin{split}
			\big(	\vartheta^v_{2n,k} \big)^2
			& = \big(g'\big(x(\tau^\sharp_{2n+1,2k})\big)\big)^2\big(\vartheta^y_{2n,k}\big)^2 + \left(g''\big(x(\tau^\flat_{2n+1,2k})\big)\right)^2\big(\wt\zeta^x_{2n+1,2k}\big)^2\\
			&\qquad + 2g'\big(x(\tau^\sharp_{2n+1,2k})\big)g''\big(x(\tau^\flat_{2n+1,2k})\big)\vartheta^y_{2n,k}\wt\zeta^x_{2n+1,2k}.	\end{split}
	\end{equation}
	for intermediate times $\tau^\sharp_{2n+1,k}, \tau^\flat_{2n+1,k} \in [2^{-2n-1}{k},2^{-2n-1}(k+1)]$. It remains to compute the contribution of each term in \eqref{three terms eq}. For simplicity, we will consider the special case $i = 0$, and analogous computations can be done for $0 \le i \le 2^n-1$. First, as $x$ is $\alpha$-H\"{o}lder continuous for $\alpha < H$, then there exists $c_x > 0$ such that 
	\begin{equation*}
		|x(\tau_{n+2,4k}(s))-x(\tau_{n+2,4k+1}(s))| \le c_x|\tau_{n+2,4k}(s) - \tau_{n+2,4k+1}(s)|^\alpha \le c_x2^{-\alpha n}.
	\end{equation*}
	The same argument also leads to 
	\begin{equation*}
		\left|x\Big(\frac{2k+1}{2^{n+1}}+s\Big)-x\Big(\frac{2k}{2^{n+1}}+s\Big)\right| \le c_x 2^{-\alpha n}
	\end{equation*}
	for $s \in [0,1]$. The above inequalities then lead to
	\begin{equation*}
		|\zeta^x_{n+1,2k}(s)| \le c_x^2 2^{(\frac{1}{2}-2\alpha)n+\frac{1}{2}} \quad \text{and} \quad  |\wt\zeta^x_{n+1,2k}| \le c_x^2 2^{(\frac{1}{2}-2\alpha)n+2}.
	\end{equation*}
	Furthermore, as $g \in C^2(\bR)$, there exists $\kappa_x > 0$ such that $32 (g''(x(s)))^2 \le \kappa_x$ for all $s\in[0,1]$. Then,
	\begin{equation}\label{eq contribution 1}
		2^{(4H-3)n}\sum_{k = 0}^{2^{n}-1}\left(g''\big(x(\tau^\flat_{2n+1,2k})\big)\right)^2\left(\wt\zeta^x_{2n+1,2k}
		\right)^2 \le \kappa_x c_x^2 2^{(4H - 8\alpha)n}.
	\end{equation}
	In addition, as $g \in C^2(\bR)$ and $x$ is continuous, the intermediate value theorem yields the existence of intermediate times $\tau^a_{n,0}, \tau^b_{n,0} \in [0,2^{-n}]$ such that
	\begin{equation}\label{eq contribution 3}
		\scalemath{0.95}{\big(g'(x(\tau^a_{n,0}))\big)^2\sum_{k = 0}^{2^n-1}\big(\vartheta^y_{2n,k}\big)^2 \le \sum_{k = 0}^{2^n-1} \big(g'\big(x(\tau^\sharp_{2n+1,2k})\big)\big)^2\big(\vartheta^y_{2n,k}\big)^2 \le \big(g'(x(\tau^b_{n,0}))\big)^2\sum_{k = 0}^{2^n-1}\big(\vartheta^y_{2n,k}\big)^2.}
	\end{equation}
	Applying \eqref{eq pathwise} then yields the existence of $n_{1,x} \in \bN$ such that for $n \ge n_{1,x}$,
	\begin{equation*}
		\begin{split}
			\big(g'(x(\tau^a_{n,0}))\big)^2  (1 - \delta_n)^2 \tau_H &\le 	2^{n(4H-3)} \sum_{k = 0}^{2^n-1} \big(g'\big(x(\tau^\sharp_{2n+1,2k})\big)\big)^2\big(\vartheta^y_{2n,k}\big)^2 \\&\le \big(g'(x(\tau^b_{n,0}))\big)^2  (1 + \delta_n)^2 \tau_H.
		\end{split}
	\end{equation*}
	Finally, by the Cauchy--Schwarz inequality, for $n \ge n_{1,x}$, we have 
	\begin{equation}\label{eq contribution 2}
		\begin{split}
			&2^{(4H-3)n}\left|\sum_{k = 0}^{2^n-1}g'\big(x(\tau^\sharp_{2n+1,2k})\big)g''\big(x(\tau^\flat_{2n+1,2k})\big)\vartheta^y_{2n,k}\wt\zeta^x_{2n+1,2k}\right|\\ \le & \,c_x \sqrt{\kappa_x\tau_H}(1 + \delta_n)2^{(2H - 4\alpha)n}\big|g'(x(\tau^b_{n,0}))\big|.
		\end{split}
	\end{equation}
	Taking $\alpha = \frac{7}{8}H$ in \eqref{eq contribution 1} and \eqref{eq contribution 2} gives 
	\begin{equation*}
		\scalemath{0.9}{	\begin{split}
				&\quad2^{(4H-3)n}\left(2\left|\sum_{k = 0}^{2^n-1}g'\big(x(\tau^\sharp_{2n+1,2k})\big)g''\big(x(\tau^\flat_{2n+1,2k})\big)\vartheta^y_{2n,k}\wt\zeta^x_{2n+1,2k}\right| + \sum_{k = 0}^{2^n-1}\left(g''\big(x(\tau^\flat_{2n+1,2k})\big)\right)^2\left(\wt\zeta_{2n+1,2k}
				\right)^2\right)\\&\le \kappa_x c_x^2 2^{-3Hn} + c_x \sqrt{\kappa_x\tau_H}(1 + \delta_n)2^{-3Hn/2+1}\big|g'(x(\tau^b_{n,0}))\big| .
		\end{split}}
	\end{equation*}
	As $\delta_n \da 0$ as $n \ua \infty$ and $g'\circ x$ is continuous, there exists $n_{2,x} \in \bN$ such that for $n \ge n_{2,x}$,
	\begin{equation*}
		\kappa_x c_x^2 2^{-3Hn} + c_x \sqrt{\kappa_x\tau_H}(1 + \delta_n)2^{-3Hn/2+1}\big|g'(x(\tau^b_{n,0}))\big| \le 2^{-5Hn/4}\big|g'(x(\tau^b_{n,0}))\big|.
	\end{equation*}
	Take $n_x:= n_{1,x} \vee n_{2,x}$. Together with \eqref{three terms eq} and \eqref{eq contribution 3}, the above inequality yields that for $n \ge n_x$, we get 
	\begin{equation*}
		\begin{split}
			&\left(\big(g'(x(\tau^a_{n,0}))\big)^2(1 - \delta_n)^2-2^{-5Hn/4}\big|g'(x(\tau^b_{n,0}))\big|\right) \tau_H \le2^{n(4H-3)} \norm{\bar{\bm \vartheta}^v_{2n,0}}^2_{\ell_2}\\  & \le \left(\big(g'(x(\tau^b_{n,0}))\big)^2(1 + \delta_n)^2+2^{-5Hn/4}\big|g'(x(\tau^b_{n,0}))\big|\right)  \tau_H.
		\end{split}
	\end{equation*}
	In particular, note that the value of $n_{x}$ depends only on the trajectory $x$ but not on the index $i$, thus, the above inequality carries over to all $0 \le i \le 2^n-1$. This completes the proof.
\end{proof}

\begin{lemma}\label{lemma bound 2}
	Suppose that $g \in C^2(\bR)$ is strictly monotone, and let $x$ be a sample path of $W^H$. Then, there exist $n_x \in \bN$, $c_H > 0$ and $\lambda_x > 0$ such that for $n \ge n_x$, 
	\begin{equation}\label{eq bound 1}
		\begin{split}
			(1 - \delta_n)^2(1 - \varepsilon_n) \tau_H \left(\int_{0}^{1}\big(g'(x(s))\big)^2\,ds\right)  &\le2^{n(4H-4)}\norm{\bar{\bm \vartheta}^v_{2n}}^2_{\ell_2}\\ &\le (1 + \delta_n)^2(1+\varepsilon_n) \tau_H \left(\int_{0}^{1}\big(g'(x(s))\big)^2\,ds\right),
		\end{split}
	\end{equation}
	where $\delta_n = c_H \cdot 2^{-n/2}\sqrt{n\log 2  + 2\log n}$ and $\varepsilon_n = \lambda_x \cdot 2^{-nH}\sqrt{n}$.
\end{lemma}
\begin{proof}
	We begin by proving the upper bound in \eqref{eq bound 1}. It follows from \Cref{lemma bound} that there exists $n_{3,x} \in \bN$ such that for $n \ge n_{3,x}$, we have
	\begin{equation}\label{eq upper bound 1}
		\begin{split}
			2^{n(4H-4)}\norm{\bar{\bm \vartheta}^v_{2n}}^2_{\ell_2} &= 2^{-n}\sum_{k = 0}^{2^n-1}2^{n(4H-3)} \norm{\bar{\bm \vartheta}^{v}_{2n,k}}^2_{\ell_2} \\&\le 2^{-n}\sum_{i = 0}^{2^n-1}\left(\big(g'(x(\tau^b_{n,i}))\big)^2(1 + \delta_n)^2+2^{-5Hn/4}\big|g'(x(\tau^b_{n,i}))\big|\right)  \tau_H.
		\end{split}
	\end{equation}
	Furthermore, it follows from \cite[Theorem 7.2.14]{MarcusRosen} that the fractional Brownian motion $W^H$ admits an exact uniform modulus of continuity $\omega(u) = u^H\sqrt{\log(1/u)}$. That is, 
	\begin{equation*}
		\mathbb{P}\left(\lim_{h \da 0} \sup_{\substack{t, s \in[0,1] \\|t-s|<h}} \frac{\left|W^H_t - W^H_s\right|}{\omega(|t-s|)}=\sqrt{2}\right)=1.
	\end{equation*}
	Hence, there exists $n_{4,x} \in \bN$ such that for $n \ge n_{4,x}$ and $0 \le i \le 2^n-1$, we have 
	\begin{equation}\label{eq modulus}
		|x(s) - x(\tau^b_{n,i})| \le \sqrt{2}\cdot \omega\left(|s - \tau^b_{n,i}|\right)\le \sqrt{2}\cdot \omega(2^{-n}) = 2^{-Hn}\sqrt{2n\log2},
	\end{equation}
	for $s \in [2^{-n}i,2^{-n}(i+1)]$. Since $g \in C^2(\bR)$ and $x \in C[0,1]$, then there exist positive constants $\kappa_x, \wt \kappa_x > 0$ such that $|g'(x(s))| \le \kappa_x$ and $|g''(x(s))| \le \wt \kappa_x$ for $s \in [0,1]$. Together with \eqref{eq modulus}, it then yields that for $n \ge n_{4,x}$, 
	\begin{equation}\label{eq sum integral}
		\begin{split}
			&\quad\,\left|\int_{0}^{1}\big(g'(x(s))\big)^2\,ds - 2^{-n}\sum_{i = 0}^{2^n-1}\big(g'(x(\tau^b_{n,i}))\big)^2\right|\\& = \left|\sum_{i = 0}^{2^n-1}\int_{2^{-n}i}^{2^{-n}(i+1)}\left(\big(g'(x(s))\big)^2 -\big(g'(x(\tau^b_{n,i}))\big)^2\right)ds \right|\\&= \left|\sum_{i = 0}^{2^n-1}\int_{2^{-n}i}^{2^{-n}(i+1)}\left(g'(x(s))+g'(x(\tau^b_{n,i}))\right)g''(x(\wt \tau^b_{n,i}))\big(x(s) - x(\tau^b_{n,i})\big)ds \right|\\& \le \kappa_x \wt \kappa_x 2^{-nH} \sqrt{8n\log2},
		\end{split}
	\end{equation}
	where $\wt \tau^b_{n,i} \in [2^{-n}i,2^{-n}(i+1)]$ are intermediate times. Take $$\lambda_x:= \sqrt{32\log2}\cdot\kappa_x\wt \kappa_x \big(\int_{0}^{1}\big(g'(x(s))\big)^2\,ds\big)^{-1},$$ and it then follows that
	\begin{equation}\label{eq upper bound 2}
		2^{-n}(1 + \delta_n)^2\sum_{i = 0}^{2^n-1}\big(g'(x(\tau^b_{n,i}))\big)^2 \le (1 + \delta_n)^2\left(1 +\frac{\lambda_x}{2}\cdot2^{-nH}\sqrt{n}\right)\int_{0}^{1}\big(g'(x(s))\big)^2\,ds.
	\end{equation}
	As $g' \circ x$ is continuous, we have $\sup_{s \in [0,1]} |g'(x(s))| < \infty$ and $\sup_n 2^{-n}\sum_{i =0}^{2^{n}-1}|g'(x(\tau^b_{n,i}))| < \infty$. Moreover, by assumption, we have $\int_{0}^{1}(g'(x(s)))^2\,ds > 0$. Finally, we have $(1 - \delta_n)^2 \ua 1$ as $n \ua \infty$. Thus, there exists $n_{5,x} \in \bN$ such that for $n \ge n_{5,x}$, we get 
	\begin{align}
		2^{-(1+5H/4)n}\sum_{i = 0}^{2^n-1}\big|g'(x(\tau^b_{n,i}))\big| &\le\frac{ \lambda_x}{2}\cdot2^{-Hn}(1 - \delta_n)^2\sqrt{n}\int_{0}^{1}\big(g'(x(s))\big)^2\,ds \label{eq lower sum}\\ &\le\frac{ \lambda_x}{2}\cdot2^{-Hn}(1 + \delta_n)^2\sqrt{n}\int_{0}^{1}\big(g'(x(s))\big)^2\,ds. \label{eq upper sum}
	\end{align}
	Thus, for $n \ge  n_{3,x} \vee n_{4,x} \vee n_{5,x}$, it follows from \eqref{eq upper bound 1}, \eqref{eq upper bound 2} and \eqref{eq upper sum} that
	\begin{equation*}
		\begin{split}
			2^{n(4H-4)}\norm{\bar{\bm \vartheta}^v_{2n}}^2_{\ell_2} &= 2^{-n}\sum_{k = 0}^{2^n-1}2^{n(4H-3)} \norm{\bar{\bm \vartheta}^{v}_{2n,k}}^2_{\ell_2}\\ &\le (1 + \delta_n)^2\left(1 + \lambda_x2^{-nH}\sqrt{n}\right)\tau_H\int_{0}^{1}\big(g'(x(s))\big)^2\,ds,
		\end{split}
	\end{equation*}
	which yields the upper bound in \eqref{eq bound 1}. 
	
	For the lower bound, note that $\tau^a_{n,i}$ are also intermediate times within $[2^{-n}i,2^{-n}(i+1)]$. Following the arguments in \eqref{eq sum integral} yields the existence of $n_{6,x} \in \bN$ such that for $n \ge n_{6,x}$,
	\begin{equation*}
		\left|\int_{0}^{1}\big(g'(x(s))\big)^2\,ds - 2^{-n}\sum_{i = 0}^{2^n-1}\big(g'(x(\tau^b_{n,i}))\big)^2\right| \le \frac{\lambda_x}{2} 2^{-nH}\sqrt{n}\int_{0}^{1}\big(g'(x(s))\big)^2\,ds,
	\end{equation*}
	which then implies
	\begin{equation*}
		2^{-n}(1 - \delta_n)^2\sum_{i = 0}^{2^n-1}\big(g'(x(\tau^a_{n,i}))\big)^2 \ge (1 - \delta_n)^2\left(1 -\frac{\lambda_x}{2}\cdot2^{-nH}\sqrt{n}\right)\int_{0}^{1}\big(g'(x(s))\big)^2\,ds.
	\end{equation*}
	This, together with \eqref{eq lower sum}, shows that for $n \ge n_{3,x} \vee n_{5,x} \vee n_{6,x}$, 
	\begin{equation*}
		2^{n(4H-4)}\norm{\bar{\bm \vartheta}^v_{2n}}^2_{\ell_2}  \ge (1 - \delta_n)^2\left(1 - \lambda_x2^{-nH}\sqrt{n}\right)\tau_H\int_{0}^{1}\big(g'(x(s))\big)^2\,ds.
	\end{equation*}
	Now, taking $n_x := n_{3,x} \vee n_{4,x} \vee n_{5,x} \vee n_{6,x}$ completes the proof.
\end{proof}

\begin{lemma}\label{lemma convergence}
	Suppose that $g \in C^2(\bR)$ is strictly monotone, and let $x$ be a typical sample path of $W^H$. Then, we have
	\begin{equation*}
		\left|\wh \cR_{2n}\left(\frac{v}{\sqrt{\tau_H \int_{0}^{1}\big(g'(x(s))\big)^2\,ds}}\right) - H\right|  = \mathcal{O}\left(n^{-\frac{1}{2}}\cdot2^{-(H \wedge \frac{1}{2})n}\right).
	\end{equation*} 
\end{lemma}
\begin{proof}
	Note that 
	\begin{equation*}
		\begin{split}
			H - \wh \cR_{2n}\left(\frac{v}{\sqrt{\tau_H \int_{0}^{1}\big(g'(x(s))\big)^2\,ds}}\right)&=(H - 1) + \frac{1}{4n}\log_2 \frac{\norm{\bar{\bm \vartheta}^v_{2n}}^2_{\ell_2}}{\tau_H \int_{0}^{1}\big(g'(x(s))\big)^2\,ds}\\&= \frac{1}{4n}\log_2 \frac{2^{n(4H-4)}\norm{\bar{\bm \vartheta}^v_{2n}}^2_{\ell_2}}{\tau_H \int_{0}^{1}\big(g'(x(s))\big)^2\,ds}\\&=  \frac{1}{4n}\log_2\left( 1 + \left(1-\frac{2^{n(4H-4)}\norm{\bar{\bm \vartheta}^v_{2n}}^2_{\ell_2}}{\tau_H \int_{0}^{1}\big(g'(x(s))\big)^2\,ds}\right)\right)\\& \sim \frac{1}{4n}\left(1-\frac{2^{n(4H-4)}\norm{\bar{\bm \vartheta}^v_{2n}}^2_{\ell_2}}{\tau_H \int_{0}^{1}\big(g'(x(s))\big)^2\,ds} \right)\quad \text{as} \quad n \ua \infty.
		\end{split}
	\end{equation*}
	Applying \Cref{lemma bound 2} gives 
	\begin{equation*}
		\frac{1}{4n}\left(1-\frac{2^{n(4H-4)}\norm{\bar{\bm \vartheta}^v_{2n}}^2_{\ell_2}}{\tau_H \int_{0}^{1}\big(g'(x(s))\big)^2\,ds} \right) \sim \frac{1}{n}\left((1 + \delta_n)^2(1 + \varepsilon_n) - 1\right) \sim \frac{\delta_n \vee \varepsilon_n}{n} \quad \text{as} \quad n \ua \infty\footnote{\text{For real-valued sequences $(a_n)$ and $(b_n)$, we write $a_n \sim b_n$ as $n \ua \infty$ if $\lim\limits_{n \ua \infty} a_n/b_n = c$ for some $c > 0$.}},
	\end{equation*}
	where $\delta_n$ and $\varepsilon_n$ are as in \Cref{lemma bound 2}. This completes the proof.
\end{proof}

\begin{proof}[Proof of \Cref{thm main}]
	It was shown in \cite[Theorem 1.4]{HanSchiedGirsanov} that the law of $(X_t)_{t \in [0,1]}$ is absolutely continuous with respect to the law of $(x_0 + W^H_t)_{t \in [0,1]}$. Hence, it suffices to prove this assertion for fractional Brownian motion $W^H$ and $V_t = \int_{0}^{t}g(W^H_s) \,ds$. 
	
	Now, suppose that $n = 2m$ for some $m \in \bN$. It then follows from \Cref{lemma convergence} that with probability one,
	\begin{equation*}
		\left|\wh \cR_{n}\left(\frac{V}{\sqrt{\tau_H \int_{0}^{1}\big(g'(W^H_s)\big)^2\,ds}}\right) - H\right| = \mathcal{O}\left(n^{-\frac{1}{2}}\cdot2^{-(\frac{H}{2} \wedge \frac{1}{4})n}\right).
	\end{equation*}
	Thus, it follows from \cite[Proposition 2.6(d)]{HanSchiedDerivative} that the assertion in \Cref{thm main} holds for the case $n = 2m$ for $m \in \bN$. For the case $n = 2m+1$ for $m \in \bN$, the assertion can be proved analogously. This completes the proof.
\end{proof}

\bibliographystyle{plain}
\bibliography{CTBook.bib}

\begin{thebibliography}{10}

\bibitem{BaudoinHairer}
Fabrice Baudoin and Martin Hairer.
\newblock A version of {H}{\"o}rmander's theorem for the fractional {B}rownian
  motion.
\newblock {\em Probability theory and related fields}, 139(3):373--395, 2007.

\bibitem{BayerGatheralFriz16}
Christian Bayer, Peter Friz, and Jim Gatheral.
\newblock Pricing under rough volatility.
\newblock {\em Quantitative Finance}, 16(6):887--904, 2016.

\bibitem{RoughvolBook}
Christian Bayer, Peter~K. Friz, Masaaki Fukasawa, Jim Gatheral, Antoine
  Jacquier, and Mathieu Rosenbaum, editors.
\newblock {\em Rough volatility}.
\newblock Financial Mathematics. Society for Industrial and Applied
  Mathematics, Philadelphia, 2024.

\bibitem{BolkoPakkanen2022GMM}
Anine Bolko, Kim Christensen, Mikko Pakkanen, and Bezirgen Veliyev.
\newblock A {GMM} approach to estimate the roughness of stochastic volatility.
\newblock {\em Journal of Econometrics}, 235(2):745--778, 2023.

\bibitem{Chong2022CLT}
Carsten~H. Chong, Marc Hoffmann, Yanghui Liu, Mathieu Rosenbaum, and Gr\'egoire
  Szymanski.
\newblock Statistical inference for rough volatility: central limit theorems.
\newblock {\em Ann. Appl. Probab.}, 34(3):2600--2649, 2024.

\bibitem{Chong2022Minimax}
Carsten~H. Chong, Marc Hoffmann, Yanghui Liu, Mathieu Rosenbaum, and Gr\'egoire
  Szymansky.
\newblock Statistical inference for rough volatility: minimax theory.
\newblock {\em Ann. Statist.}, 52(4):1277--1306, 2024.

\bibitem{ContDasArtefact}
Rama Cont and Purba Das.
\newblock Rough volatility: fact or artefact?
\newblock {\em Sankhya B}, pages 1--33, 2024.

\bibitem{EuchRosenbaumGatheral2019roughening}
Omar El~Euch, Jim Gatheral, and Mathieu Rosenbaum.
\newblock Roughening {H}eston.
\newblock {\em Risk}, pages 84--89, 2019.

\bibitem{EuchRosenbaumRoughHeston18}
Omarl El~Euch and Mathieu Rosenbaum.
\newblock Perfect hedging in rough {H}eston models.
\newblock {\em The Annals of Applied Probability}, 28(6):3813--3856, 2018.

\bibitem{Forde2022}
Martin Forde, Masaaki Fukasawa, Stefan Gerhold, and Benjamin Smith.
\newblock {The Riemann--Liouville field and its GMC as H→ 0, and skew
  flattening for the rough Bergomi model}.
\newblock {\em Statistics \& Probability Letters}, 181:109265, 2022.

\bibitem{Fukasawa2022Estimation}
Masaaki Fukasawa, Tetsuya Takabatake, and Rebecca Westphal.
\newblock Consistent estimation for fractional stochastic volatility model
  under high-frequency asymptotics.
\newblock {\em Mathematical Finance}, 32(4):1086--1132, 2022.

\bibitem{Gatheral2018rough}
Jim Gatheral, Thibault Jaisson, and Mathieu Rosenbaum.
\newblock A rough volatility model.
\newblock {\em Quantitative Finance}, 18(6):933--949, 2018.

\bibitem{GatheralRosenbaum}
Jim Gatheral, Thibault Jaisson, and Mathieu Rosenbaum.
\newblock Volatility is rough.
\newblock {\em Quantitative Finance}, 18(6):933--949, 2018.

\bibitem{HanSchiedDerivative}
Xiyue Han and Alexander Schied.
\newblock Estimating the roughness exponent of stochastic volatility from
  discrete observations of the realized variance.
\newblock {\em arXiv 2307.02582}, 2023.

\bibitem{HanSchiedGirsanov}
Xiyue Han and Alexander Schied.
\newblock A criterion for absolute continuity relative to the law of fractional
  {B}rownian motion.
\newblock {\em Electronic Communications in Probability}, 29:1--10, 2024.

\bibitem{HanSchiedMatrix}
Xiyue Han and Alexander Schied.
\newblock {Robust Faber--Schauder approximation based on discrete observations
  of an antiderivative}.
\newblock {\em \rm{To appear in} Mathematics of Operations Research}, 2025.

\bibitem{HanSchiedHurst}
Xiyue Han and Alexander Schied.
\newblock The roughness exponent and its model-free estimation.
\newblock {\em The Annals of Applied Probability}, 35(2):1049--1082, 2025.

\bibitem{MarcusRosen}
Michael~B. Marcus and Jay Rosen.
\newblock {\em Markov processes, {G}aussian processes, and local times}, volume
  100 of {\em Cambridge Studies in Advanced Mathematics}.
\newblock Cambridge University Press, Cambridge, 2006.

\end{thebibliography}

\clearpage
\end{document}